\documentclass[letterpaper, 10 pt, conference]{ieeeconf}  

\IEEEoverridecommandlockouts                              
\newcommand{\norm}[1]{\left\lVert#1\right\rVert}

\usepackage{graphics} 
\usepackage{epsfig} 
\usepackage{mathptmx} 
\usepackage{times} 
\usepackage{amsmath} 
\usepackage{amssymb}  
\usepackage{mathtools}
\usepackage{algorithm}
\usepackage{dsfont}
\usepackage{algorithmicx}
\usepackage{algpascal}
\usepackage{textcomp}

\newtheorem{lm}{Lemma}
\newtheorem{rmk}{Remark}
\newtheorem{thm}{Theorem}

\title{\LARGE \bf
Error Bounds for Kernel-Based Linear System Identification with Unknown Hyperparameters
}
\author{Mingzhou Yin, Roy S. Smith
\thanks{This work was supported by the Swiss National Science Foundation under Grant 200021\_178890.}
\thanks{The authors are with the Automatic Control Laboratory, Swiss Federal Institute of Technology (ETH Zurich), Physikstrasse 3, 8092 Zurich, Switzerland, {\tt\small \{myin,rsmith\}} {\tt\small @control.ee.ethz.ch}.}%
}

\begin{document}

\maketitle
\thispagestyle{empty}
\pagestyle{empty}

\begin{abstract}
The kernel-based method has been successfully applied in linear system identification using stable kernel designs. From a Gaussian process perspective, it automatically provides probabilistic error bounds for the identified models from the posterior covariance, which are useful in robust and stochastic control. However, the error bounds require knowledge of the true hyperparameters in the kernel design and are demonstrated to be inaccurate with estimated hyperparameters for lightly damped systems or in the presence of high noise. In this work, we provide reliable quantification of the estimation error when the hyperparameters are unknown. The bounds are obtained by first constructing a high-probability set for the true hyperparameters from the marginal likelihood function and then finding the worst-case posterior covariance within the set. The proposed bound is proven to contain the true model with a high probability and its validity is verified in numerical simulation.
\end{abstract}

\section{Introduction}
System identification estimates models of dynamical systems from input-output data. Under the assumption that a low-dimensional model structure is known \textit{a priori}, the model parameters can be estimated by maximum likelihood estimation, of which one example is the prediction error method (PEM) \cite{Ljung_2002}. This framework, despite some numerical difficulties, has been very successful in various applications \cite{LjungBook2}.

However, as more complex and large-scale systems are emerging, low-dimensional model structures become less accessible. Following the seminal work \cite{PILLONETTO201081}, system identification can be reformulated as a non-parametric function learning problem, and solved using, among other approaches, the kernel-based method \cite{Chiuso_2019,Ljung2_2019,pillonetto2022regularized}. The kernel-based method can be interpreted as function learning in a reproducing kernel Hilbert space (RKHS), Gaussian process (GP) regression, or ridge regression with basis expansions. In linear system identification, a truncated impulse response model is identified with a weighted $l_2$ regularization term, with a class of stable kernels designed to identify stable linear systems effectively \cite{PILLONETTO201081,Pillonetto_2016,CHEN2018109}. The GP interpretation provides the kernel-based method with one of its main advantages: it obtains Gaussian stochastic models and thus high-probability error bounds simultaneously with the nominal model \cite{Chen_2012}. This enables its application in robust and stochastic control.

However, one often-neglected aspect of kernel-based identification is that the results, including the error bounds, are conditioned on correct hyperparameter selection, in the same way as PEM is conditioned on the correct model structure. The hyperparameters in the kernel-based method are usually selected separately using the maximum marginal likelihood method or cross-validation, and used in identification empirically with certainty equivalence. This makes the GP-based error bounds unreliable when the estimated hyperparameters are inaccurate, and thus detrimental to use in safety-critical applications. This phenomenon has been observed in machine learning literature \cite{10.7551/mitpress/3206.001.0001}. In linear system identification, it is demonstrated in this paper that the error bound derived from estimated hyperparameters can be inaccurate, especially for lightly damped systems and in low signal-to-noise ratio scenarios.

Therefore, error bounds for the kernel-based methods are needed in the case of unknown hyperparameters. In kernel-based linear system identification, \cite{Pillonetto_2022} establishes non-asymptotic bounds for all stable systems with bounded pole magnitudes. However, the bounds are too conservative and thus only useful for sample complexity analysis. Error bounds are also widely studied in GP literature, e.g., \cite{Maddalena_2021,Srinivas_2012,NEURIPS2019_fe73f687,pmlr-v70-chowdhury17a}, which are usually obtained by scaling posterior standard deviations. Such bounds are derived under the assumptions that the prior covariance function is exact and/or an upper bound of the RKHS-induced norm is known. However, both assumptions are impractical in general. Several works provide modified bounds considering the discrepancy between the estimated and the true hyperparameters \cite{pmlr-v162-capone22a,beckers2018mean,Fiedler_2021,10.5555/3455716.3455903}. These results depend on prior knowledge of the magnitude of the discrepancy: a set of possible kernels, bounds on the kernel function values, and Mat\'{e}rn kernel parameter bounds are required in \cite{beckers2018mean,Fiedler_2021,10.5555/3455716.3455903} respectively. Such information is estimated from data in \cite{pmlr-v162-capone22a} by investigating the maximum marginal likelihood problem in hyperparameter estimation. Unfortunately, these works all consider approximation problems instead of regression problems and often consider kernel classes that do not contain the typical stable kernels used in linear system identification.

In this work, we provide probabilistic error bounds for kernel-based linear system identification with no prior knowledge of the hyperparameters by extending \cite{pmlr-v162-capone22a} to regression problems and stable kernels. The proof in \cite{pmlr-v162-capone22a} is also simplified with an improved constant. This approach assumes the correct kernel structure and first estimates a high-probability set for the hyperparameters from the marginal likelihood function. Then, the upper bound of the posterior covariance is found within the range of hyperparameters. A uniform bound is obtained for diagonal and tuned/correlated kernels. For general kernels, element-wise bounds can be found by optimization. Finally, probabilistic error bounds are established by scaling the worst-case posterior standard deviations. Optimization problems to obtain the tightest error bounds are discussed as well. Numerical simulations confirm that the proposed error bounds correctly bound the estimation error with a high probability.

\textit{Notation.} A Gaussian distribution with mean $\mu$ and covariance $\Sigma$ is indicated by $\mathcal{N}(\mu,\Sigma)$. The expectation and the covariance of a random vector $x$ are denoted by $\mathbb{E}(x)$ and $\text{cov}(x)$, respectively. The symbol $\overset{p}{\leq}$ indicates less than or equal to with probability $p$. For a vector $x$, the weighted $l_2$-norm $(x^\top Px)^{\frac{1}{2}}$ is denoted by $\norm{x}_P$. The $(i,j)$-th element and the trace of a matrix $A$ are denoted by $A_{i,j}$ and $\text{tr}(A)$, respectively.

\section{The Kernel-Based Method in Linear System Identification}

\subsection{Problem Statement}
Consider a causal and stable linear time-invariant single-input single-output discrete-time system $y_t=G(q)u_t+v_t$, where $u_t$, $y_t$, $v_t$ are the inputs, outputs, and additive noise respectively, and $q$ is the shift operator. The additive noise is assumed to be zero-mean i.i.d. Gaussian with a variance of $\sigma^2$. An input-output sequence of the system
\begin{equation}
    \mathbf{u}=\left[u_1\ u_2\ \dots\ u_N\right]^\top\!,\  \mathbf{y}=\left[y_1\ y_2\ \dots\ y_N\right]^\top
\end{equation}
has been collected. We are interested in identifying the transfer function $G(q)$ from the data sequence $(\mathbf{u},\mathbf{y})$. In this work, we consider the finite impulse response model of $G(q)$: $G(q)=\sum_{l=0}^{n_g-1}g_l q^{-l}$, i.e., 
\begin{equation}
    y_t=\sum_{l=0}^{n_g-1}g_l u_{t-l} + v_t.
\end{equation}
This leads to the data equation:
\begin{equation}
\underbrace{
    \begin{bmatrix}
    y_1\\y_2\\\vdots\\y_N
    \end{bmatrix}}_{\mathbf{y}} = \underbrace{\begin{bmatrix}
    u_1&0&\cdots&0\\
    u_2&u_1&\cdots&0\\
    \vdots&\vdots&\ddots&\vdots\\
    u_N&u_{N-1}&\cdots&u_{N-n_g+1}\\
    \end{bmatrix}}_{\Phi}
    \underbrace{
    \begin{bmatrix}
    g_0\\g_1\\\vdots\\g_{n_g-1}
    \end{bmatrix}}_\mathbf{g} + 
    \underbrace{\begin{bmatrix}
    e_1\\e_2\\\vdots\\e_N
    \end{bmatrix}}_{\mathbf{e}}.
\end{equation}
Here, we assume that the system starts at rest before the data collection, i.e., $u_t=0$, $\forall t\leq 0$.

\subsection{The Kernel-Based Method}

If no prior knowledge is assumed for $G(q)$, the maximum likelihood estimator of $\mathbf{g}$ is given by the least-squares solution 
\begin{equation}
\hat{\mathbf{g}}^\text{LS}=\text{arg}\underset{\mathbf{g}}{\text{min}}\ \norm{\mathbf{y}-\Phi\,\mathbf{g}}_2^2=\left(\Phi^\top\Phi\right)^{-1}\Phi^\top \mathbf{y}.
\end{equation}
It is well known that the estimation error is also Gaussian with covariance $\text{cov}(\mathbf{g})=\sigma^2\left(\Phi^\top\Phi\right)^{-1}=:\Sigma^\text{LS}$. Element-wise stochastic error bounds can be obtained for $\hat{\mathbf{g}}^\text{LS}$ as
\begin{equation}
\mathbb{P}\left(\left|\hat{g}^\text{LS}_l-g_l\right|\leq \mu_\delta\sqrt{\Sigma^\text{LS}_{l,l}}\right)\geq 1-\delta,
\label{eqn:lsb}
\end{equation}
where 
\begin{equation}
    F_\mathcal{N}(\mu_\delta)\geq 1-\delta/2,
    \label{eqn:mud}
\end{equation}
$F_\mathcal{N}(\cdot)$ is the cumulative distribution function of the Gaussian distribution.
 
If the system is assumed to be stable with a low McMillan degree, the impulse response of the system will be smooth and exponentially converge to zero. This prior knowledge can be encoded as either a prior distribution in GP regression or an RKHS in kernel regression. In both cases, the estimator is given by the ridge-regularized least-squares solution 
\begin{align}
    \hat{\mathbf{g}}&=\text{arg}\underset{\mathbf{g}}{\text{min}}\ \norm{\mathbf{y}-\Phi\,\mathbf{g}}_2^2+\sigma^2\,\mathbf{g}^\top K^{-1}\mathbf{g}\\
    &=\left(\Phi^\top\Phi+\sigma^2 K^{-1}\right)^{-1}\Phi^\top \mathbf{y},
    \label{eqn:gcl}
\end{align}
with $K$ having different meanings.

In the GP regression interpretation, $K$ is the covariance of the prior distribution of $\mathbf{g}$: $\mathbf{g}\sim \mathcal{N}(\mathbf{0},K)$. Then $\mathbf{y}$ and $\mathbf{g}$ are jointly Gaussian:
\begin{equation}
    \begin{bmatrix}
        \mathbf{y}\\\mathbf{g}
    \end{bmatrix}\sim\mathcal{N}\left(
    \mathbf{0},\begin{bmatrix}
        K&K\Phi^\top\\
        \Phi K&\Phi K\Phi^\top + \sigma^2\mathbb{I}
    \end{bmatrix}
    \right).
\end{equation}
The estimate $\hat{\mathbf{g}}$ comes directly from the posterior mean of $\mathbf{g}$ given $\mathbf{y}$: $\mathbf{g}|\mathbf{y}\sim\mathcal{N}\left(\hat{\mathbf{g}},\Sigma\right)$, where $\Sigma=\sigma^2\left(\Phi^\top\Phi+\sigma^2 K^{-1}\right)^{-1}$ is the posterior covariance \cite{Chen_2012}. From the posterior covariance, the associated element-wise stochastic error bounds are
\begin{equation}
\mathbb{P}\left(\left|\hat{g}_l-g_l\right|\leq \mu_\delta\sqrt{\Sigma_{l,l}}\right)\geq 1-\delta,
\label{eqn:seb}
\end{equation}
conditioned on the prior distribution of $\mathbf{g}$.

In the kernel regression interpretation, the continuous-time impulse response function $g(t): [0,+\infty)\rightarrow\mathbb{R}$, $g(l)=g_l$, $l=0,\dots,n_g-1$ is identified by solving the regularized function learning problem within an RKHS $\mathcal{H}(k(\cdot,\cdot))$ associated with a kernel function $k(\cdot,\cdot):[0,+\infty)\times[0,+\infty)\rightarrow\mathbb{R}$:
\begin{equation}
    \begin{split}
    g^\star(\cdot)=\text{arg}\underset{g(\cdot)\in\mathcal{H}}{\text{min}}&\ \norm{\mathbf{y}-\Phi\,\mathbf{g}}_2^2+\sigma^2\norm{g(\cdot)}_\mathcal{H}^2\\
    \text{s.t.}&\ \ \mathbf{g}=\begin{bmatrix}
        g(0)\\\vdots\\g(n_g-1)
    \end{bmatrix},
    \end{split}
    \label{eqn:kernel}
\end{equation}
where $\norm{g(\cdot)}_\mathcal{H}$ is the induced norm of $g(\cdot)$ with respect to $\mathcal{H}$. From the representer theorem \cite{10.1007/3-540-44581-1_27}, the solution to \eqref{eqn:kernel} is $g^\star(x)=\mathbf{k}_x\left(\Phi^\top\Phi K+\sigma^2 \mathbb{I}\right)^{-1}\Phi^\top \mathbf{y}$ and $\mathbf{g}^\star=\hat{\mathbf{g}}$, where $K$ evaluates the kernel function associated with the RKHS $\mathcal{H}$ at $l=0,\dots,n_g-1$, i.e., $K_{l,l}=k(l,l)$, and $\mathbf{k}_x=\left[k(x,0)\ \dots\ k(x,n_g-1)\right]$. The induced norm of $g^\star$ is calculated as $\norm{g^\star(\cdot)}_\mathcal{H}^2=\hat{\mathbf{g}}^\top K^{-1}\hat{\mathbf{g}}$.

\subsection{Kernel Design and Hyperparameter Selection}

The matrix $K$ is critical to the performance of the kernel-based method. Extensive studies have been conducted to obtain appropriate structures of $K$ that promote impulse response estimates that are both smooth and exponentially converge to zero \cite{CHEN2018109}. The most commonly used kernels in linear system identification include:
\begin{enumerate}
    \item diagonal (DI): $K_{i,i}^\text{DI}(\eta)=c\lambda^i$, $K_{i,j}^\text{DI}=0$ for $i\neq j$,
    \item tuned/correlated (TC): $K_{i,j}^\text{TC}(\eta)=c\lambda^{\max(i,j)}$,
    \item stable spline (SS): $$K_{i,j}^\text{SS}(\eta)=c\lambda^{2\max(i,j)}\left(\tfrac{\lambda^{\min(i,j)}}{2}-\tfrac{\lambda^{\max(i,j)}}{6}\right),$$
\end{enumerate}
where $\eta\in \left\{[c\ \lambda]^\top \big|\,c\geq 0,0\leq\lambda\leq 1\right\}=:\mathbb{H}$ are the hyperparameters. These kernel designs have been shown effective both theoretically and numerically \cite{pillonetto2022regularized}.

The hyperparameters need to be selected before applying the estimator \eqref{eqn:gcl}. The most widely-used approach to hyperparameter selection is the maximum marginal likelihood method, which maximizes the probability of observing $\mathbf{y}$ given the inputs $\mathbf{u}$ and the hyperparameters $\eta$: $$\hat{\eta}=\text{arg}\underset{\eta}{\text{min}}\,-\log\,p(\mathbf{y}|\mathbf{u},\eta),$$ where 
\begin{equation}
\resizebox{\columnwidth}{!}{$
    p(\mathbf{y}|\mathbf{u},\eta)=\exp\left(-\frac{1}{2}\log\det\,\Psi(\eta) - \frac{1}{2}\,\mathbf{y}^\top \Psi^{-1}(\eta)\mathbf{y} + \text{const.}\right)$}
    \label{eqn:pyeta}
\end{equation}
and $\Psi=\sigma^2\mathbb{I}+\Phi K(\eta)\Phi^\top$. The estimated hyperparameters $\hat{\eta}$ are used, with certainty equivalence, to construct $K$ and then to obtain the estimate $\hat{\mathbf{g}}$.

\section{Error Bounds with Unknown Hyperparameters}

\subsection{Pitfalls with Error Bounds from Posterior Covariances}
\label{sec:3a}

The kernel-based method has shown remarkable performance in linear system identification, in terms of the nominal estimate \eqref{eqn:gcl}. However, the stochastic error bound \eqref{eqn:seb} assumes an exact prior distribution with true hyperparameters. In this work, the kernel structure is assumed to be correct with true hyperparameters $\eta_0$, but $\eta_0$ is unknown. When the estimated hyperparameters $\hat{\eta}$ are used, directly using \eqref{eqn:seb} to provide a stochastic model for safety-critical applications can be problematic.

To demonstrate this issue, consider two second-order systems
\begin{equation*}
    G_1(q)=\frac{0.4888}{q^2-1.8q+0.9^2},\ G_2(q)=\frac{0.0616}{q^2-q+0.9^2},
\end{equation*}
with two different noise levels $\sigma^2=0.1$ and 0.5. Both systems have two poles of magnitude 0.9: $G_1(q)$ has two real poles at 0.9; $G_2(q)$ has a pair of complex poles with a real part of 0.5. The systems have been normalized to have an $\mathcal{H}_2$-norm of 1.

Stochastic models given by the error bound \eqref{eqn:seb} are analyzed by 100 Monte Carlo simulations with TC kernels. Unit Gaussian inputs are used to generate the identification data. Figure~\ref{fig:1} shows the empirical probabilities of the error bounds containing the true impulse responses with $\delta=0.1$ and identification parameters $N=200$ and $n_g=50$. It can be seen that except for the case of $G_1(q)$ with low noise, the magnitudes of the errors are significantly underestimated in the other three cases, with the empirical probabilities much smaller than the target value of $1-\delta=0.9$. This indicates that the error bounds based on estimated hyperparameters are not reliable in cases where the impulse response is lightly damped and/or the signal-to-noise ratio is poor.
\begin{figure}
    \centering
    \includegraphics[width=0.9\linewidth]{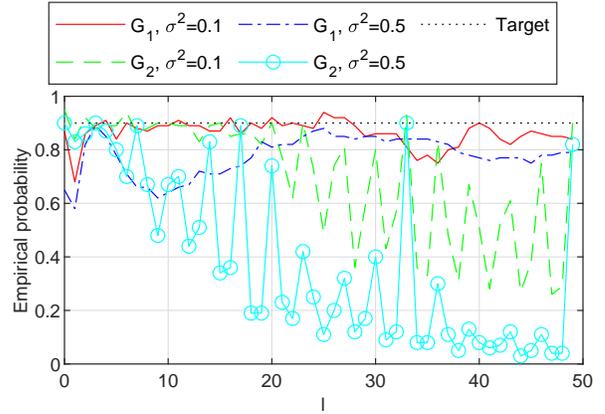}
    \vspace{-1em}
    \caption{Empirical probability of error bounds containing the true parameters using estimated hyperparameters.}
    \label{fig:1}
\end{figure}

To investigate the reason why the error bounds are inaccurate under these cases, Figure~\ref{fig:2} shows the marginal probability density $p(\mathbf{y}|\mathbf{u},\eta)$ \eqref{eqn:pyeta} with respect to the hyperparameters in one representative simulation. It can be seen that in cases (b), (c), and (d), where the error bounds based on estimated hyperparameters are inaccurate as shown in Figure~\ref{fig:1}, the marginal probability density is not strongly localized. This suggests that there could exist a large discrepancy between $\eta_0$ and $\hat{\eta}$, which leads to the misspecification of the error bounds.
\begin{figure}
    \centering
    \includegraphics[width=\linewidth]{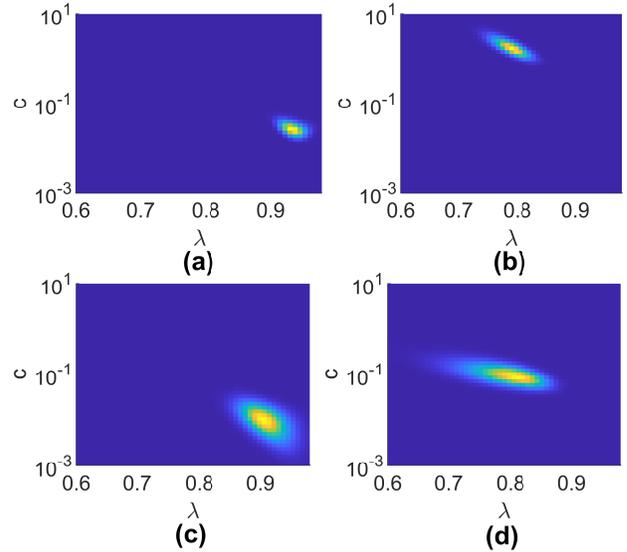}
    \caption{Marginal probability density with respect to hyperparameters. (a) $G_1,\sigma^2=0.1$, (b) $G_2,\sigma^2=0.1$, (c) $G_1,\sigma^2=0.5$, (d) $G_2,\sigma^2=0.5$.}
    \label{fig:2}
\end{figure}

\subsection{Worst-Case Posterior Variances}

To solve the problem of quantifying error bounds with unknown hyperparameters,
we first bound the true hyperparameters using the measured data. The distribution of hyperparameters conditioned on the measured data is given by
\begin{equation}
    p(\eta|\mathbf{u},\mathbf{y})=\frac{p(\mathbf{y}|\mathbf{u},\eta)p(\eta)}{\int_{\eta\in\mathbb{H}}p(\mathbf{y}|\mathbf{u},\eta)p(\eta)\,\text{d}\eta},
    \label{eqn:peta}
\end{equation}
where $p(\mathbf{y}|\mathbf{u},\eta)$ is given in \eqref{eqn:pyeta}. If prior information on the hyperparameters $\eta$ is available, it can be encoded as a hyperprior $p(\eta)$. Otherwise, $p(\eta)$ is chosen as a constant. This leads to
\begin{equation}
    \mathbb{P}\left(\eta_0\in[\eta_1,\eta_2]\right)=\frac{\int_{\eta\in[\eta_1,\eta_2]}p(\mathbf{y}|\mathbf{u},\eta)p(\eta)\,\text{d}\eta}{\int_{\eta\in\mathbb{H}}p(\mathbf{y}|\mathbf{u},\eta)p(\eta)\,\text{d}\eta}=:1-\delta',
    \label{eqn:bound}
\end{equation}
where $\eta_i=[c_i\ \lambda_i]^\top$, $i=1,2$ and 
$$[\eta_1,\eta_2]=\left\{\eta=[c\ \lambda]^\top\left|\,c_1\leq c\leq c_2,\lambda_1\leq \lambda\leq \lambda_2\right.\right\}$$
is a rectangular set. By choosing a small $\delta'$, \eqref{eqn:bound} establishes a high-probability set for the true hyperparameters. It is assumed that the set also contains the estimated hyperparameter, i.e., $\hat{\eta}\in[\eta_1,\eta_2]$.

Then, we investigate the effect of hyperparameters on the stochastic model, to find the worst-case posterior variances $\Sigma_{l,l}$ within the set $[\eta_1,\eta_2]$. For DI and TC kernels, a uniform bound is derived analytically using the following lemma.
\begin{lm}
The matrix inequality $\Sigma(\eta_1)\preccurlyeq \Sigma(\eta_2)$ is satisfied when $\left(\frac{\lambda_2}{\lambda_1}\right)^\gamma c_1\leq c_2$, $\lambda_1\leq\lambda_2$, with $\gamma=0$ for DI kernels and $\gamma=-1/\ln{\lambda_2}-1$ for TC kernels.
\label{lm:1}
\end{lm}
\begin{proof}
See Appendix~\ref{sec:ap1}.
\end{proof}

From Lemma~\ref{lm:1}, we have
\begin{equation}
    \Sigma(\eta_0)\overset{1-\delta'}{\leq}\sigma^2 \left(\Phi^\top\Phi+\sigma^2 \left(\frac{\lambda_1}{\lambda_2}\right)^{\gamma} K^{-1}(\eta_2)\right)^{-1}=:\bar{\Sigma}.
    \label{eqn:sigl10}
\end{equation}
So the posterior variances with true hyperparameters $\eta_0$ can be uniformly bounded by
\begin{equation}
    \Sigma_{l,l}(\eta_0)\overset{1-\delta'}{\leq} \bar{\Sigma}_{l,l}=:\sigma_l^2.
    \label{eqn:sigl1}
\end{equation}
Since $\hat{\eta}\in[\eta_1,\eta_2]$, we also have $\Sigma_{l,l}(\hat{\eta})\leq\sigma_l^2$.

For a general kernel structure, the bound can be computed element-wise by solving the optimization problem:
\begin{equation}
    \sigma_l^2=\max_{\eta\in[\eta_1,\eta_2]}\ \Sigma_{l,l}(\eta).
    \label{eqn:sigl2}
\end{equation}

\subsection{Stochastic Error Bounds}
\label{sec:3c}

We are now ready to present the main result of the paper.
\begin{thm}
The impulse response estimate \eqref{eqn:gcl} with estimated hyperparameters $\hat{\eta}$ admits the following stochastic element-wise error bound:
\begin{equation}
    \mathbb{P}\left(\left|\hat{g}_l(\hat{\eta})-g_l\right|\leq \bar{\mu}\sigma_l\right)\geq (1-\delta)(1-\delta'),
\end{equation}
where $\bar{\mu}=\mu_\delta+\frac{2}{\sigma}\norm{\mathbf{y}}_S$ and $S=\Phi\left(\Phi^\top\Phi\right)^{-1}\Phi^\top$. Variables $\mu_\delta$, $\delta$, and $\delta'$ are given in \eqref{eqn:mud} and \eqref{eqn:bound}.
\label{thm:1}
\end{thm}
\begin{proof}
The estimation error is decomposed as
\begin{align}
\left|\hat{g}_l(\hat{\eta})-g_l\right|&\,\,\leq\ \left|\hat{g}_l(\hat{\eta})-\hat{g}_l(\eta_0)\right|+\left|\hat{g}_l(\eta_0)-g_l\right|\\
&\overset{1-\delta}{\leq}\left|\hat{g}_l(\hat{\eta})-\hat{g}_l(\eta_0)\right|+\mu_\delta\sqrt{\Sigma_{l,l}(\eta_0)},\label{eqn:bound0}
\end{align}
where the two terms are due to misspecified hyperparameters and measurement noise respectively.

Define the posterior kernel
\begin{equation*}
    \resizebox{\columnwidth}{!}{$k^p_\eta(x,x')=k_\eta(x,x')-\mathbf{k}_x(\eta)\left(K(\eta)+\sigma^2\left(\Phi^\top\Phi\right)^{-1}\right)^{-1}\mathbf{k}_x(\eta)^\top$}.
\end{equation*}
Note that $k^p_\eta(i,j)=\Sigma_{i,j}(\eta)$. The associated RKHS is denoted as $\mathcal{H}^p_\eta$. It is easy to see that $g^\star_\eta(\cdot)\in\mathcal{H}^p_\eta$ and $\norm{g^\star_\eta(\cdot)}_{\mathcal{H}^p_\eta}^2=\hat{\mathbf{g}}^\top(\eta)\Sigma^{-1}(\eta)\hat{\mathbf{g}}(\eta)$. Note the reproducing property of the RKHS $g^\star_\eta(x)=\langle g^\star_\eta(\cdot), k^p_\eta(\cdot,x)\rangle_{\mathcal{H}^p_\eta}$, where $\langle\cdot,\cdot\rangle_{\mathcal{H}^p_\eta}$ denotes the inner product in ${\mathcal{H}^p_\eta}$. From the Cauchy–Schwarz inequality, we have $\left|g^\star_\eta(x)\right|\leq k^p_\eta(x,x)^{\frac{1}{2}}\norm{g^\star_\eta(\cdot)}_{\mathcal{H}^p_\eta}$. This leads to
\begin{equation}
\begin{split}
    \!\!\!\left|\hat{g}_l(\eta)\right|^2&\leq\Sigma_{l,l}(\eta)\hat{\mathbf{g}}^\top(\eta)\Sigma^{-1}(\eta)\hat{\mathbf{g}}(\eta)\\
    &=\frac{1}{\sigma^2}\Sigma_{l,l}(\eta)\mathbf{y}^\top\Phi\left(\Phi^\top\Phi+\sigma^2 K^{-1}(\eta)\right)^{-1}\!\!\Phi^\top \mathbf{y}\\
    &\leq\frac{\Sigma_{l,l}(\eta)}{\sigma^2}\norm{\mathbf{y}}_S^2.    
\end{split}
\label{eqn:bound2}
\end{equation}
So we have $\left|\hat{g}_l(\hat{\eta})\right|^2\leq \frac{\sigma_l^2}{\sigma^2}\norm{\mathbf{y}}_S^2$ and $\left|\hat{g}_l(\eta_0)\right|^2\overset{1-\delta'}{\leq} \frac{\sigma_l^2}{\sigma^2}\norm{\mathbf{y}}_S^2$. Then,
\begin{align}
    \left|\hat{g}_l(\hat{\eta})-\hat{g}_l(\eta_0)\right|\leq \left|\hat{g}_l(\hat{\eta})\right|+ \left|\hat{g}_l(\eta_0)\right|\overset{1-\delta'}{\leq} \frac{2\sigma_l}{\sigma}\norm{\mathbf{y}}_S.
    \label{eqn:bound1}
\end{align}

From \eqref{eqn:bound}, \eqref{eqn:sigl1}, \eqref{eqn:sigl2}, we have $\mu_\delta\sqrt{\Sigma_{l,l}(\eta_0)}\overset{1-\delta'}{\leq}\mu_\delta\sigma_l$. This, together with \eqref{eqn:bound0} and \eqref{eqn:bound1}, proves Theorem~\ref{thm:1}.
\end{proof}
\begin{rmk}
For DI and TC kernels, by modifying the last inequality in \eqref{eqn:bound2}, the bound in Theorem~\ref{thm:1} can be tightened by choosing $S=\Phi\left(\Phi^\top\Phi+\sigma^2 \left(\frac{\lambda_1}{\lambda_2}\right)^{\gamma}K^{-1}(\eta_2)\right)^{-1}\Phi^\top$.
\end{rmk}

\subsection{Selecting the Set of Hyperparameters}
Theorem~1 holds for any choices of $\eta_1,\eta_2$ that satisfy \eqref{eqn:bound}. To obtain the tightest bound, $\eta_1,\eta_2$ can be selected by optimization. For DI and TC kernels, the total magnitude of the bounds $\sum_{l=0}^{n_g-1}\bar{\mu}\sigma_l$ can be minimized. From \eqref{eqn:sigl10} and \eqref{eqn:sigl1}, this is equivalent to solving
\begin{align}
    \underset{\eta_1,\eta_2}{\text{min}}&\quad\left(\tfrac{\lambda_2}{\lambda_1}\right)^{\gamma}\text{tr}\left(K(\eta_2)\right)\\
    \text{s.t.}&\quad\frac{\int_{\eta\in[\eta_1,\eta_2]}p(\mathbf{y}|\mathbf{u},\eta)p(\eta)\,\text{d}\eta}{\int_{\eta\in\mathbb{H}}p(\mathbf{y}|\mathbf{u},\eta)p(\eta)\,\text{d}\eta}\geq 1-\delta'\label{eqn:conprob}
\end{align}
For a general kernel structure with element-wise bound \eqref{eqn:sigl2}, $\eta_1,\eta_2$ can be selected individually for each $l$ by solving the minimax problem:
\begin{align}
    \sigma_l^2=\underset{\eta_1,\eta_2}{\text{min}}\ \underset{\eta\in[\eta_1,\eta_2]}{\text{max}}\ \Sigma_{l,l}(\eta)\quad\text{s.t.}\ \ \eqref{eqn:conprob}.
    \label{eqn:etaopt}
\end{align}

\section{Numerical Results}
\label{sec:4}

The proposed bound is verified numerically by considering the same examples as in Section~\ref{sec:3a}. The error bound \eqref{eqn:seb} with estimated hyperparameters analyzed in Section~\ref{sec:3a} is termed the \textit{vanilla kernel bound}, whereas the proposed bound in Section~\ref{sec:3c} is called the \textit{robust kernel bound}. The \textit{least-squares bound} \eqref{eqn:lsb} is also compared.

For computational efficiency, the optimization problems to find $\eta_1,\eta_2$ are solved by discretizing $\eta$. The nominal estimate and the estimated hyperparameters are obtained by \texttt{impulseest} in \textsc{Matlab}. The inner problem in \eqref{eqn:etaopt} is solved by \texttt{fmincon} in \textsc{Matlab}. For the robust kernel bound, we select $\delta'=0.1$ and $\bar{\mu}=\mu_\delta$,  since the theoretical constant in Theorem~\ref{thm:1} is too conservative in practice. Such conservativeness is often observed in GP error bounds, so a much smaller scaling factor is often selected in practical applications \cite{NIPS2017_766ebcd5,Umlauft_2017}.

Figure~\ref{fig:3} presents a comparison of the performance of different error bounds with a TC kernel design. For each example, the left figure shows representative identification results in one simulation, whereas the right figure shows the empirical probability of error bounds containing the true parameters from 100 Monte Carlo simulations. The results show that the proposed robust kernel bounds are more conservative compared to the vanilla kernel bounds, especially under high noise, but they are much more reliable with much higher empirical probabilities of containing the true parameters. On the other hand, the robust kernel bounds are still much tighter than the least-squares bounds.

\begin{figure*}
    \centering
    \includegraphics[width=0.9\linewidth]{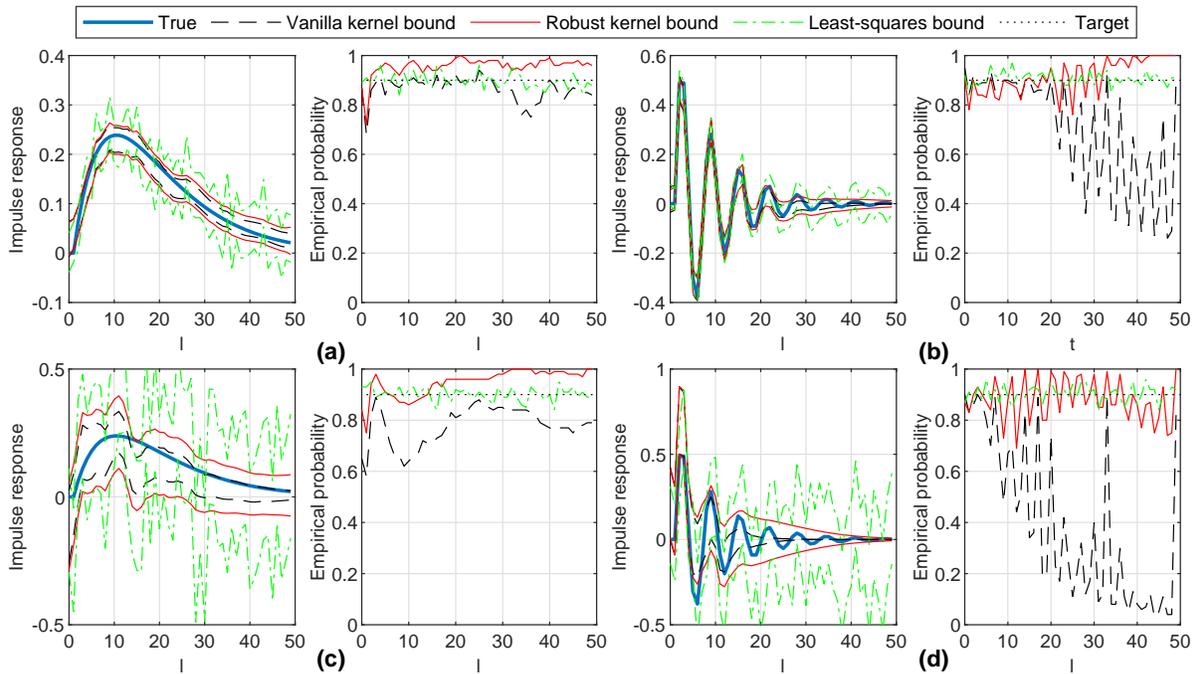}
    \caption{Comparison of different error bounds with TC kernels. (a) $G_1,\sigma^2=0.1$, (b) $G_2,\sigma^2=0.1$, (c) $G_1,\sigma^2=0.5$, (d) $G_2,\sigma^2=0.5$. Left: representative identification results, right: the empirical probability of error bounds containing the true parameters.}
    \label{fig:3}
\end{figure*}

Figure~\ref{fig:4} shows the empirical probability with a SS kernel design. The robust kernel bounds are derived by selecting $\sigma_l$ from \eqref{eqn:etaopt}. Similar results to the TC kernel case are obtained, where the robust kernel bounds are much more reliable than the vanilla kernel bounds.

\begin{figure}
    \centering
    \includegraphics[width=0.9\linewidth]{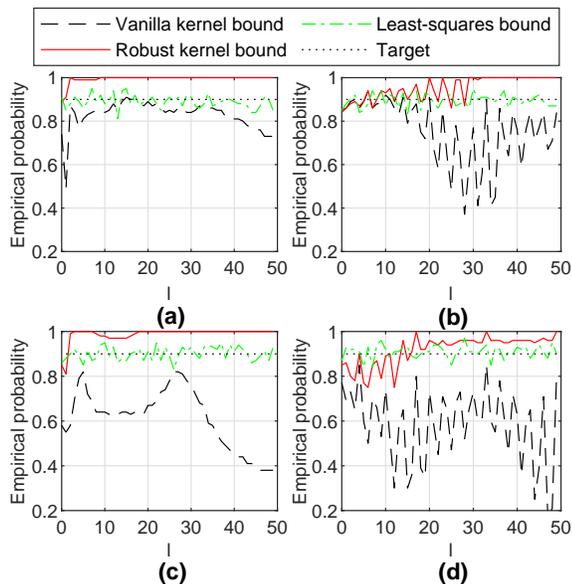}
    \vspace{-1em}
    \caption{Empirical probability of error bounds containing the true parameters with SS kernels. (a) $G_1,\sigma^2=0.1$, (b) $G_2,\sigma^2=0.1$, (c) $G_1,\sigma^2=0.5$, (d) $G_2,\sigma^2=0.5$. }
    \label{fig:4}
\end{figure}

\section{Conclusions}
In this work, we investigate the problem of quantifying the estimation error in kernel-based linear system identification with unknown hyperparameters. First, it is illustrated that the certainty equivalence principle does not work here: error bounds constructed using the estimated hyperparameters are too optimistic in multiple examples. Instead, a rectangular set of hyperparameters is constructed to contain the true ones with high probability. The error bounds can then be obtained by scaling the worst-case posterior variances within the set. It is shown both theoretically and numerically that the proposed bound is accurate in specifying the estimation error.

This work provides a practical approach to obtaining a reliable stochastic model centered around the nominal estimate of kernel-based system identification. Further research directions include deriving uniform posterior covariance bounds for other kernel structures and improving the constant $\bar{\mu}$ in Theorem~\ref{thm:1}.

\appendices

\section{Proof of Lemma~\ref{lm:1}}
\label{sec:ap1}
The result is trivial for DI kernels. For TC kernels, define $M(\mathbf{m}_n)\in\mathbb{R}^{n\times n}$, $\mathbf{m}_n=\left[m_1\ m_2\ \dots\ m_n\right]^\top$ with $M_{i,j}(\mathbf{m}_n)=m_{\max(i,j)}$. We first prove that
\begin{equation}
    \det\,M(\mathbf{m}_n)=m_n\prod_{i=1}^{n-1}\left(m_i-m_{i+1}\right).
    \label{eqn:M}
\end{equation}

For $n=1,2$, $\det\,M(\mathbf{m}_1)=m_1$, $\det\,M(\mathbf{m}_2)=m_2(m_1-m_2)$ satisfy \eqref{eqn:M}. Suppose \eqref{eqn:M} is satisfied for $n=l-1,l$. Define
$$\mathbf{m}_{l\setminus i}=\left[m_1\ \dots\ m_{i-1}\ m_{i+1}\ \dots\ m_l\right]^\top.$$
From the definition of the determinant, we have
$
    \det\,M(\mathbf{m}_l)=m_l\sum_{i=1}^{l}(-1)^{l-i}\det\,M(\mathbf{m}_{l\setminus i})
$.

For $n=l+1$,
\begin{equation*}
\resizebox{\columnwidth}{!}{%
$\begin{aligned}&\det\,M(\mathbf{m}_{l+1})=m_{l+1}\sum_{i=1}^{l+1}(-1)^{l+1-i}\det\,M(\mathbf{m}_{l+1\setminus i})\\
    =&\ m_{l+1}\left(\det\,M(\mathbf{m}_l)+\sum_{i=1}^{l}(-1)^{l+1-i}\det\,M(\mathbf{m}_{l+1\setminus i})\right)\\
    =&\ m_{l+1}\left(\det\,M(\mathbf{m}_l)-\frac{m_{l+1}}{m_l}(m_{l}-m_{l+1})\sum_{i=1}^{l-1}(-1)^{l-i}\det\,M(\mathbf{m}_{l\setminus i})\right.\\
    &\qquad\qquad\qquad\qquad\qquad\left.-\frac{m_{l+1}}{m_{l-1}}(m_{l-1}-m_{l+1})\det\,M(\mathbf{m}_{l-1})\right)\\
    =&\ m_{l+1}\left(\det\,M(\mathbf{m}_l)-\frac{m_{l+1}}{m_l}(m_{l}-m_{l+1})\det\,M(\mathbf{m}_l)\right.\\
    &\qquad\qquad\qquad\qquad\qquad\left.-\frac{m_{l+1}^2(m_{l-1}-m_l)}{m_l m_{l-1}}\det\,M(\mathbf{m}_{l-1})\right)\\
    =&\ m_{l+1}\left(1-\frac{m_{l+1}(m_{l}-m_{l+1})}{m_l}-\frac{m_{l+1}^2}{m_l^2}\right)\det\,M(\mathbf{m}_l)\\
    =&\ m_{l+1}\prod_{i=1}^{l}\left(m_i-m_{i+1}\right)\end{aligned}$}
\end{equation*}
satisfies \eqref{eqn:M}. This proves \eqref{eqn:M} by induction.

Using Sylvester's criterion, $M(\mathbf{m}_n)$ is positive semidefinite iff $\det\,M(\mathbf{m}_l)\geq 0, \forall\,l=1,\dots,n$. This requires
\begin{equation}
    m_i-m_{i+1}\geq 0, \forall\,i=1,\dots,n-1.
    \label{eqn:mii}
\end{equation}
Define $\eta'_2=\left[\left(\frac{\lambda_2}{\lambda_1}\right)^\gamma c_1\ \lambda_2\right]^\top$. Since $\left(\frac{\lambda_2}{\lambda_1}\right)^\gamma c_1\leq c_2$, we have $K(\eta_2)\succcurlyeq K(\eta'_2)$. Define $M(\mathbf{m}_{n_g})=K(\eta'_2)-K(\eta_1)$ by choosing $m_i=\left(\frac{\lambda_2}{\lambda_1}\right)^\gamma c_1\lambda_2^i-c_1\lambda_1^i$. So $K(\eta'_2)-K(\eta_1)\succcurlyeq 0$ is equivalent to
\begin{equation}
\lambda_2^{1+\gamma}-\lambda_1^{1+\gamma}\geq\lambda_2^{2+\gamma}-\lambda_1^{2+\gamma}\geq\dots\geq\lambda_2^{n_g+\gamma}-\lambda_1^{n_g+\gamma}.
\label{eqn:lambda}
\end{equation}
Note that $f(x)=\lambda_2^x-\lambda_1^x$ is monotonically non-increasing for $x\geq -1/\ln{\lambda_2}$, $\forall\,\lambda_2\geq\lambda_1$. This indicates that \eqref{eqn:lambda} is satisfied for $\gamma\geq -1/\ln{\lambda_2}-1$. Therefore, $K(\eta_2)\succcurlyeq K(\eta'_2)\succcurlyeq K(\eta_1)$ for $\gamma=-1/\ln{\lambda_2}-1$, which leads to
$$\left(\Phi^\top\Phi+\sigma^2 K^{-1}(\eta_2)\right)^{-1}\succcurlyeq \left(\Phi^\top\Phi+\sigma^2 K^{-1}(\eta_1)\right)^{-1}.$$
This directly proves Lemma~\ref{lm:1}.


\bibliographystyle{IEEEtran}
\bibliography{refs}

\end{document}